\documentclass[11pt]{article}
\usepackage{amsmath,amsfonts,amssymb,amsthm}
\usepackage{fullpage}
\usepackage[colorlinks]{hyperref}
\hypersetup{linkcolor=cyan,filecolor=cyan,citecolor=cyan,urlcolor=cyan}
\usepackage{xspace}
\usepackage{thm-restate,color,xcolor}
\usepackage{boxedminipage}
\usepackage[boxed]{algorithm}
\usepackage{epigraph}
\usepackage[sc]{mathpazo}
\usepackage{amsmath}
\usepackage{mathtools}
\usepackage{framed}
\usepackage[framemethod=tikz]{mdframed}
\usepackage{titlesec}
\usepackage{lipsum}%
\usepackage{cleveref,aliascnt}
\usepackage{tikz}
\usepackage{wrapfig}
\usepackage{framed}
\usepackage[framemethod=tikz]{mdframed} 

\usetikzlibrary{shapes.geometric}
\usetikzlibrary{arrows}
\usetikzlibrary{arrows.meta}
\usetikzlibrary{patterns}
\usetikzlibrary{shapes.misc}

\newcommand{\congest}{${\mathsf{CONGEST}}$}

\newcommand{\dist}{\mbox{\rm dist}}

\theoremstyle{plain}
\newtheorem{thm}{Theorem}[section]
\newcommand{\BTHM}{\begin{thm}} \newcommand{\ETHM}{\end{thm}}
\newtheorem{cor}[thm]{Corollary}
\newcommand{\BCR}{\begin{cor}} \newcommand{\ECR}{\end{cor}}
\newtheorem{lem}[thm]{Lemma}
\newcommand{\BL}{\begin{lem}}   \newcommand{\EL}{\end{lem}}
\newtheorem{clm}[thm]{Claim}
\newcommand{\BCM}{\begin{clm}}   \newcommand{\ECM}{\end{clm}}
\newtheorem{prop}[thm]{Proposition}
\newcommand{\BP}{\begin{prop}}   \newcommand{\EP}{\end{prop}}
\newtheorem{assm}[thm]{Assumption}
\newcommand{\BASM}{\begin{assm}}   \newcommand{\EASM}{\end{assm}}

\theoremstyle{definition}
\newtheorem{defn}{Definition}[section]
\newcommand{\BD}{\begin{defn}}   \newcommand{\ED}{\end{defn}}
\newtheorem{con}[thm]{Conjecture}
\newcommand{\BCONJ}{\begin{con}}   \newcommand{\ECONJ}{\end{con}}

\theoremstyle{definition}
\newtheorem{problem}[thm]{Problem}
\newcommand{\BPR}{\begin{problem}}   \newcommand{\EPR}{\end{problem}}

\newenvironment{rem}{\noindent{\bf Remark:~~}}{}
\newcommand{\BREM}{\begin{rem}} \newcommand{\EREM}{\end{rem}}
\newenvironment{discussion}{\noindent{\bf Discussion:~~\\}}{}
\newcommand{\BDIS}{\begin{discussion}} \newcommand{\EDIS}{\end{discussion}}
\newtheorem{obs}{Observation}[section]
\numberwithin{equation}{section}

\def\blackslug
     {\hbox{\hskip 1pt\vrule width 8pt height 8pt depth 1.5pt\hskip 1pt}}
\def\qed{\quad\blackslug\lower 8.5pt\null\par}

\newcommand{\dilation}{\mbox{\tt d}}
\newcommand{\congestion}{\mbox{\tt c}}

\newtheorem{exmp}[thm]{Example}
\newtheorem{fact}[thm]{Fact}
\newcommand{\BEX}{\begin{exmp}} \newcommand{\EEX}{\end{exmp}}
\newcommand{\BF}{\begin{fact}}   \newcommand{\EF}{\end{fact}}
\newcommand{\Bcr}{\begin{techcorr}}
\newcommand{\Ecr}{\end{techcorr}}
\newcommand{\BDS}{\begin{description}}
\newcommand{\EDS}{\end{description}}
\newcommand{\BE}{\begin{enumerate}}
\newcommand{\EE}{\end{enumerate}}
\newcommand{\BI}{\begin{itemize}}
\newcommand{\EI}{\end{itemize}}
\renewenvironment{proof}{\noindent{\bf Proof:~~}}{\qed}
\newcommand{\BPF}{\begin{proof}}
\newcommand{\EPF}{\end{proof}}

\newcommand{\BB}{\begin{enumerate}}
\newcommand{\EB}{\end{enumerate}}

\usepackage{geometry}
\title{Low-Congestion Shortcuts in Constant Diameter Graphs}
\author{
 Shimon Kogan\\
  \small Weizmann Institute\\
  \small shimon.kogan@weizmann.ac.il
\and
Merav Parter\\
        \small Weizmann Institute \\
        \small merav.parter@weizmann.ac.il\thanks{Supported by the European Research Council (ERC) No. 949083, and by the Israeli Science Foundation (ISF) No. 2084/18.}
}

\begin{document}
\maketitle

\begin{abstract}
Low congestion shortcuts, introduced by Ghaffari and Haeupler (SODA 2016), provide a unified framework for global optimization problems in the \congest\ model of distributed computing. Roughly speaking, for a given graph $G$ and a collection of vertex-disjoint connected subsets $S_1,\ldots, S_\ell \subseteq V(G)$, $(\congestion,\dilation)$ low-congestion shortcuts augment each subgraph $G[S_i]$ with a subgraph $H_i \subseteq G$ such that: (i) each edge appears on at most $\congestion$ subgraphs (congestion bound), and (ii) the diameter of each subgraph $G[S_i] \cup H_i$ is bounded by $\dilation$ (dilation bound). It is desirable to compute shortcuts of small congestion and dilation as these quantities capture the round complexity of many 
global optimization problems in the $\mathsf{CONGEST}$ model. For $n$-vertex graphs with constant diameter $D=O(1)$, Elkin (STOC 2004) presented an (implicit) shortcuts lower bound with\footnote{As usual, $\tilde{O}()$ and $\tilde{\Omega}()$ hide poly-logarithmic factors.} $\congestion+\dilation=\widetilde{\Omega}(n^{(D-2)/(2D-2)})$. A nearly matching upper bound, however, was only recently obtained for $D \in \{3,4\}$ by Kitamura et al. (DISC 2019). 

In this work, we resolve the long-standing complexity gap of shortcuts in constant diameter graphs, originally posed by Lotker et al. (PODC 2001). We present new shortcut constructions which match, up to poly-logarithmic terms, the lower bounds of Elkin. As a result, we provide improved and existentially optimal algorithms for several network optimization tasks in constant diameter graphs, including MST, $(1+\epsilon)$-approximate minimum cuts and more. 
\end{abstract}

\maketitle

\section{Introduction}
Low congestion shortcut is a combinatorial graph structure introduced by Ghaffari and Haeupler \cite{GhaffariH16} in the context of distributed network optimization. Specifically, low congestion shortcuts provide a unified framework for obtaining existentially nearly-tight algorithms for a large collection of global graph problems in the $\mathsf{CONGEST}$ model of distributed computing \cite{Peleg:2000}. In this model, the network is abstracted as an $n$-node graph $G=(V, E)$, with one processor on each network node. Initially, these processors do not know the graph, and they solve the given graph problem via communicating with their neighbors in a synchronous manner. Per round, each processor can send one $O(\log n)$-bit message to each of its neighboring processors.  
Low congestion shortcuts are formally defined as follows:
\\
\begin{mdframed}[hidealllines=true,backgroundcolor=gray!25]
\vspace{-5pt}
\BD [Ghaffari and Haeupler \cite{GhaffariH16}]
Given a graph $G = (V, E)$ and a collection $\mathcal{S}=\{S_1,\ldots, S_\ell\}$
of vertex-disjoint and connected subsets of $V$, a $(\dilation, \congestion)$-\emph{shortcut} of $G$ and $\mathcal{S}$ is defined by a set of subgraphs $\mathcal{H} = \{H_1, H_2, \ldots, H_\ell\}$ of $G$ such that: 
\begin{enumerate}
\item For each $i$, the diameter of $G[S_i] \cup H_i$ is at most $\dilation$.
\item Each edge $e \in E$ appears on at most $\congestion$ subgraphs $\{G[S_i] \cup H_i, i\in \{1,\ldots, \ell\}\}$.
\end{enumerate}
\ED
\end{mdframed}
In other words, the efficiency of the shortcuts is characterized by two parameters: the \emph{dilation} measured by the maximum diameter $\dilation$ over all subgraphs, and the \emph{congestion} measured by the largest number $\congestion$ of augmented subgraphs that use a given edge. The summation of the dilation and congestion is usually referred to as the \emph{quality} of the shortcuts. In their influential work, Ghaffari and Haeupler \cite{GhaffariH16} observed that the round complexity of the Minimum Spanning Tree (MST) and the approximate minimum-cut problems in the $\mathsf{CONGEST}$ model can be bounded (up to poly-logarithmic terms) by the quality of the shortcuts. This in particular implies that any round complexity lower bound for the MST problem implies also a lower bound on the quality of the shortcuts. Low-congestion shortcuts have been proven useful for a wide collection of tasks, including the computation of approximate shortest-paths \cite{HaeuplerL18}, DFS trees \cite{GhaffariP17}, graph diameter \cite{LiP19}, and the approximation of minimum weight two-edge connected subgraphs \cite{dory2019improved}. 

For any $n$-vertex graph of (unweighted) diameter $D$, Ghaffari and Haeupler \cite{GhaffariH16} observed the existence of shortcuts with quality $O(D+\sqrt{n})$. This bound is known to be nearly tight by the MST lower bound results of Elkin \cite{Elkin04} and Das-Sarma et al. \cite{sarma2012distributed}. As shortcuts fully capture the round complexity of many graph problems, there has been a great effort in characterizing graph families for which shortcuts with a considerably improved quality can be obtained. A representative list of these families includes: planar graphs \cite{GhaffariH16,haeupler2016LCSwoEmbed}, graphs with excluded minors \cite{HaeuplerLZ18,GhaffariHMinorArxiv}, highly connected graphs \cite{ChuzhoyPT20}, expander graphs \cite{GhaffariKS17}, graphs with bounded chordality and clique-width \cite{KitamuraKOI19}. In all of these works, the improved shortcuts can also be computed in a round complexity that matches their quality, thus providing improved algorithms for MST and other optimization problems, for these graph families.  

One notable graph family that has attracted a significant amount of attention for more than two decades, concerns the family of \emph{constant diameter} graphs. It has been widely noted that many of the real-world networks have a very small diameter, independent of the number of network's participants. In the context of social networks, such as the Facebook, this phenomenon is usually explained by the ``six degrees of separations". The diameter of the world-wide web, as another example, is bounded by $19$ while having billions of nodes (pages) \cite{albert1999diameter}. 
This apparent ubiquity of constant diameter networks motivated the design of improved algorithms for this graph family. The canonical problem in this regard is MST.

%
%
For graphs of diameter $D=1$ (i.e., the complete graphs), Lotker et al. \cite{lotker2003mst} presented an $O(\log\log n)$-round MST algorithm. Following a sequence of improvements \cite{HegemanPPSS15,GhaffariP16}, the state-of-the-art round complexity of the problem is $O(1)$ rounds \cite{Jurdzinski018,NowickiMST20}. For graphs with diameter $D=2$, Lotker et al. \cite{LotkerPP01,LotkerPP06} presented MST algorithms with $O(\log n)$ rounds, and for graphs with $D \in \{3,4\}$, they gave a lower bound of $\widetilde{\Omega}(n^{1/4})$ and $\widetilde{\Omega}(n^{1/3})$ rounds, respectively \footnote{Since the MST complexity for $D \geq 3$ is already polynomial in $n$, poly-logarithmic terms are hidden.}. This in turn also implies a lower bound of $\widetilde{\Omega}(n^{1/4}), \widetilde{\Omega}(n^{1/3})$ on the quality of the shortcuts in graphs with diameter $3,4$.
In a seminal paper, Elkin \cite{Elkin04} extended the MST lower-bound of Lotker et al. for any constant diameter graphs. In particular, their result, stated in the shortcut terminology, implies the existence of an $n$-vertex $D$-diameter graph $G^*$, and specific vertex disjoint connected subsets $\mathcal{S}=\{S_1,\ldots, S_\ell\}$, such that any $(\congestion,\dilation)$ shortcuts for $\mathcal{S}$ must satisfy that  $\congestion+\dilation=\widetilde{\Omega}(n^{(D-2)/(2D-2)})$.  For the special case of $D \in\{3,4\}$, recently Kitamura et al. \cite{KitamuraKOI19} presented an upper bound construction of shortcuts, which matches the lower bound of Lotker et al. \cite{LotkerPP01,LotkerPP06}. For graphs with diameter of $D \geq 5$, shortcuts with quality $o(\sqrt{n})$ were not known to this date.

\paragraph{Our Results:} In this paper, we resolve the long-standing open problem regarding the complexity of MST computation in constant diameter graphs for any $D \geq 5$. More generally, we settle the complexity on the quality of low-congestion shortcuts in these graphs, matching the $\widetilde{\Omega}(n^{(D-2)/(2D-2)})$ lower bound by Elkin \cite{Elkin04} and Das-Sarma et al. \cite{SarmaHKKNPPW11,sarma2012distributed}. Our key result is: \\
\begin{mdframed}[hidealllines=true,backgroundcolor=gray!25]
\vspace{-5pt}
\BTHM\label{thm:main-result}
Let $D \geq 3$ be a constant. For any graph $G$ on $n$ vertices and of diameter $D$, there exists a randomized distributed algorithm for computing low-congestion $(\congestion,\dilation)$ shortcuts with quality $\congestion+\dilation=\widetilde{O}(k_D)$ in $\widetilde{O}(k_D)$ rounds where $k_D = n^{(D-2)/(2D-2)}$, w.h.p.\footnote{As usual, w.h.p. refers to a success guarantee of $1-1/n^c$ for any input constant $c>1$.}. 
\ETHM
\end{mdframed}
It is noteworthy that our approach is completely independent of the shortcut algorithm by Kitamura et al. \cite{KitamuraKOI19} for the case of $D=4$. Our shortcut construction is based on a very simple procedure (similar to the algorithm of Kitamura et al. for $D=3$ \cite{KitamuraKOI19}) where each edge is sampled into a shortcut subgraph (of sufficiently many nodes) with some fixed probability. A sampling-based approach for low congestion shortcuts has been also applied by Ghaffari, Kuhn and Su \cite{GhaffariKS17} for the family of expander graphs. 
The congestion bound of our shortcuts follows immediately by a simple application of the Chernoff bound. Our main efforts are devoted for providing a new analysis for bounding the diameter (i.e., dilation) of each augmented subgraph. We introduce the concept of shortcut trees, auxiliary graphs which allows us to analyze the dilation of the construction, by applying a recursive argument on the shortcuts introduced on each $s$-$t$ shortest path in $G[S_j]$. We note that our approach for bounding the dilation using the concept of shortcut trees is the main technical contribution of this paper. Using the improved shortcuts, we obtain a collection of improved and existentially tight algorithms for any $n$-vertex graph of constant diameter $D=O(1)$. \\
\begin{mdframed}[hidealllines=true,backgroundcolor=gray!25]
\vspace{-5pt}
\BCR[Distributed MST and $(1+\epsilon)$ Approximate Minimum Cut]\label{cor:MST}
For every $n$-vertex graph with diameter $D=O(1)$, there is a randomized distributed algorithm for computing MST and $(1+\epsilon)$ approximation of the minimum cut in $\widetilde{O}(n^{(D-2)/(2D-2)})$ rounds. These bounds are nearly existentially tight by  \cite{SarmaHKKNPPW11,sarma2012distributed}. 
\ECR
\end{mdframed}
Our improved shortcuts also have various immediate applications for additional problems, such as approximate SSSP \cite{HaeuplerL18} and $O(\log n)$-approximation of the minimum weight two-edge connected subgraphs \cite{dory2019improved}.

There are two interesting open ends to our construction. The first is concerned with the message complexity. The total message complexity of our shortcut algorithm is bounded by $\widetilde{O}(m n^{(D-2)/(2D-2)})$. It will be interesting to improve this bound to $\widetilde{O}(m)$ messages. An additional aspect is concerned with a derandomization of our construction. These aspects have been settled for general graphs by Haeupler, Hershkowitz and Wajc \cite{HaeuplerHW18}.

\section{The Low-Congestion Shortcut Algorithm}
We start by presenting the centralized construction of the shortcuts, and then explain the distributed implementation. Let $\mathcal{S}=\{S_1,\ldots, S_\ell\}$ be a collection of connected node-disjoint subsets in $G$, and define
$$k_D=n^{(D-2)/(2D-2)}~\mbox{~and~} N=\lceil n/k_D \rceil~.$$
A subset $S_i$ is said to be \emph{small} if $|S_i|\leq k_D$, and otherwise it is \emph{large}. Clearly, it is sufficient to compute shortcut subgraphs for at most $N$ large subsets. For ease of notation, let $S_{1},\ldots, S_{N}$ be the large subsets in $\mathcal{S}$. We start by considering the case where the diameter $D$ is even, and towards the end explain the minor modifications required to handle odd diameters, as well.
\\
\begin{mdframed}[hidealllines=true,backgroundcolor=gray!25]
\vspace{-3pt}
\textbf{Centralized Shortcut Construction:} 
For every $i \in \{1,\ldots, N\}$ compute the subgraph $H_i$ as follows:
\begin{enumerate}
\item Each node $v \in S_{i}$ adds all its incident edges to $H_i$.
\item Each node $u \in V \setminus S_i$ adds each of incident edge $(u, v)$ to $H_i$ independently with probability $\mathbf{p}=(k_D \cdot \log n)/N$. 
\item \textbf{Repeat} Step (2) for $D$ (independent) times.
\end{enumerate}
\end{mdframed}
Note that in this description, each edge $\{u,v\}$ is sampled in a directed manner, where $(u,v)$ (resp., $(v,u)$) is sampled by the endpoint $u$ (resp., $v$) $D$ times into $H_i$ independently with probability $\mathbf{p}$. 

\paragraph{The congestion argument.} We show that the congestion of the subgraphs $H_1,\ldots, H_N$ is $O(D \cdot k_D \cdot \log n)$, w.h.p. Note that since the subsets $S_{i}$ pairwise disjoint, the congestion introduced by Step (1) of the algorithm is bounded\footnote{Since each edge $(u,v)$ is added (at most) to the subsets of $u$ and $v$.} by 2. Consider now the congestion induced by the edges added in Step (2). Every edge $\{u,v\}$ is sampled by both of its endpoints $u$ into $2D \cdot N \cdot \mathbf{p}=O(K_D\cdot \log n)$ subgraphs, in expectation. Thus by a simple application of the Chernoff bound, we get that each (directed) edge gets sampled into at most $O(k_D \cdot \log n)$ subgraphs, w.h.p. 

The heart of the analysis for bounding the diameter of each $G[S_i] \cup H_i$ by $O(k_D \log n)$ appears in Subsection \ref{sec:dilation}. We next provide a distributed implementation which mimics the above mentioned centralized construction using $\widetilde{O}(k_D)$ rounds in the \congest\ model of distributed computing. 

\paragraph{Distributed implementation.} Following \cite{GhaffariH16}, the input to the distributed construction assumes that each part $S_i \in \mathcal{S}$ is identified by the identifier of the node $v_i$ of maximum ID in $S_i$. At the beginning of the algorithm, all nodes in $S_i$ know the ID of $v_i$, and in the output of the algorithm, each node in $V$ knows its incident edges in each $G[S_i] \cup H_i$. 

Note that the nodes are required to know $k_D$ which is a function of the exact diameter $D$ and the number of nodes $n$. The exact knowledge of $n$ and a $2$-factor approximation of the diameter can be obtained within $O(D)$ rounds by computing a BFS tree from an arbitrary node. We first describe the construction assuming that all nodes know exactly $D$, and thus $k_D$, and then explain how to omit this assumption.  

\noindent \textbf{Shortcuts construction assuming the knowledge of $D$.} First, the algorithm identifies the collection of large subsets and number them in a sequential manner in $[1,N]$. This can be done by applying the following $O(k_D)$-round procedure. Compute a (possibly) truncated BFS tree rooted at $v_i$ of depth at most $k_D$ is computed in parallel in each graph $G[S_i]$ for every $i$. As a result, we obtain a $k_D$-depth BFS tree rooted at each node $v_i$, which allows $v_i$ to determine if its component $S_i$ is large (e.g., if the tree is not spanning all nodes in $S_i$). Using additional $O(D)$ rounds, the nodes can also number the large components from $[1,N]$, that is, the leader $v_i$ of each large component $S_i$ knows its index $i$. 

Next, as all nodes know the number $N$, they can locally compute their edges in each $H_i$ for $i \in \{1,\ldots, N\}$. Specifically, each node $u$ samples each edge $(u,v)$ into $H_i$ by sampling the edge $D$ time independently with probability of $\mathbf{p}$. The edge $(u,v)$ is taken into $H_i$ if at least one of these sampling steps is successful. Our next goal is to compute a (possibly truncated) BFS tree in each $G[S_i] \cup H_i$ rooted at the node $v_i$. The depth of each tree is restricted to $\widetilde{O}(k_D)$. By the distributed input to the problem, all nodes in $S_i$ know $v_i$. However, nodes not in $S_i$, but possibly in $H_i$, might not know it but rather only the index $i$. Recall that it is desired that each edge $(u,v)$ will know the identifier $ID(v_i)$ of each shortcut subgraph $H_i$ to which it belongs. This is obtained by applying the following $\widetilde{O}(k_D)$-round procedure.

All $N$ (possibly truncated) BFS trees in $G[S_1] \cup H_1, \ldots, G[S_N] \cup H_N$ are computed in parallel using the \emph{random delay} approach \cite{leighton1999fast,ghaffari2015near}:
\BTHM[{\cite[Theorem 1.3]{ghaffari2015near}}]\label{thm:delay}
Let $G$ be a graph and let $A_1,\ldots,A_m$ be $m$ distributed algorithms in 
the \congest model, 
where each algorithm takes at most $\dilation$ rounds, and where for each 
edge of $G$, at most $\congestion$ messages need to go through it, in total 
over all these algorithms. Then, there is a randomized distributed
algorithm (using only private randomness) that, with high probability, produces 
a schedule that runs all the algorithms in $O(\congestion +\dilation \cdot \log 
n)$ rounds, after $O(\dilation \log^2 n)$ rounds of pre-computation.
\ETHM

In our context, the congestion and dilation bounds of each of the $N$ sub-algorithms is set to $O(k_D \cdot \log n)$. 
We also restrict the overall running time to at most $O(k_D \log^2 n)$ rounds.
For simplicity, we assume that all the nodes have an access to a string $\mathcal{SR}$ of shared randomness, where $\mathcal{SR}[i]$ for $i \in \{1, \dots, N\}$ describes a random value in the range $\{1, \dots, k_D\}$ which specifies the starting phase of the $i^{th}$ sub-algorithm. \cite{ghaffari2015near} showed that this string can be made of $O(\log^2 n)$ random bits, which can be sent to all nodes in $O(D+\log n)$ rounds. 

We divide time into phases of $O(\log n)$ rounds, and delay the start of each BFS $T_i$, rooted in $v_i$, in the subgraph $G[S_i] \cup H_i$ by a random delay of $t_i=\mathcal{SR}[i]$. Once a BFS algorithm starts at node $v_i \in S_i$, the related BFS in $G[S_i] \cup H_i$ grows at a synchronous speed of one hop per phase. Since we restrict each edge to participate in at most $O(k_D \log n)$ many $H_i$'s subgraphs, w.h.p., per phase and per edge $e'=(v, u)$, there are at most $O(\log n)$ BFS tokens scheduled to go through edge $e'$ from $v$ to $u$, in this phase. Each edge $(u,v) \in H_i$ learns the identity of $v_i$ at the time at which the BFS token of $v_i$ arrives this edge. 
The collection of the $O(k_D\log n)$-depth (possibly truncated) BFS trees in $G[S_i] \cup H_i$ for $i \in \{1,\ldots, N\}$ can be all computed in $O(k_D \cdot \log^2 n)$ rounds. 
\\
\\
\noindent \textbf{Omitting the assumption on knowing $D$.}
Recall that by computing a BFS tree in $G$, all nodes obtain a $2$-factor approximation $D'$ for the graph diameter. 
Our approach is based on guessing the diameter starting with the lowest guess $D'/2$ to $D'$. The algorithm terminates for the smallest value $D''$ for which low-congestion shortcuts with quality $O(k_{D''}\log n)$ are computed. Towards that goal, we slightly modify the above mentioned algorithm so that given a diameter estimate $D''$, where possibly $D' \leq D$, the algorithm is restricted to run in only $O(k_{D''}\cdot \log^2 n)$ rounds. At the end of the algorithm, the nodes learn whether shortcuts with quality $O(k_{D''}\log n)$ have been successfully computed or not. In the positive case, the algorithm terminates and otherwise it proceeds to the next guess $D''+1$.
 The correctness will then follow from the correctness of the shortcut algorithm for the correct diameter value $D$.

 It remains to explain how to verify if the algorithm has successfully computed shortcuts of quality $O(k_{D''}\log n)$. Let $H'_1,\ldots, H'_\ell$ be the (possibly incomplete) shortcuts computed for $S_1,\ldots, S_\ell$. 
We enforce the congestion and dilation of these shortcuts to $O(k_{D''}\log n)$ as follows. First, when letting each node $(u,v)$ sample its edges into the $H'_i$ subgraphs (of the at most $N''=\lceil n/k_{D''}\rceil$ large components), the construction terminates if the edge congestion exceeds the allowed value of $O(k_{D''}\log n)$. In addition, in each large components, the nodes compute a (possibly truncated) BFS tree of depth $O(k_{D''}\log n)$. 

To verify these shortcuts, upon computing the (possibly truncated) BFS trees in parallel, the leader $v_i$ of each component $S_i$ determines if the computed BFS tree  $T'_i$ in $G[S_i] \cup H'_i$ indeed spans all nodes in $S_i$. The guess $D'$ is considered to be \emph{successful} only if all nodes $v_i$ for $i \in \{1,\ldots, \ell\}$ have determined that their shortcut construction is complete. Since $k_D$ is an increasing function in $D$, we have that the total running time of this guessing-based algorithm is bounded by $D \cdot O(k_D \log^2 n)=O(k_D \log^2 n)$. 
In the remaining part of the paper, we focus on showing the dilation argument, i.e., proving that the diameter of each augmented graph $G[S_i] \cup H_i$ is at most $\widetilde{O}(k_D)$.

\section{Dilation Argument}\label{sec:dilation}
The structure of this section is as follows. We start by providing the high level structure of the argument for a fixed augmented graph $H=G[S_j] \cup H_j$ for some $j \in \{1,\ldots, N\}$. 
In Subsection \ref{sec:shortcut-tree} we introduce the notion of \emph{shortcut trees} which serve as our key analytical tool for bounding the diameter of $H$. Then in Subsection \ref{sec:comp-arg}, we provide the detailed dilation argument. For every $k \in \{1,\ldots, D\}$, let $E_{k}$ be the set of edges sampled into $H$ in the $k^{th}$ application of the edge sampling of Step (2) in the centralized computation of the shortcut subgraph $H$.

\BTHM\label{thm:dilation}
W.h.p., the diameter of $H$ is bounded by $O(k_D \log n)$. 
\ETHM

\paragraph{High-Level Description of the Argument.} We fix a node pair $s, t \in G[S_j]$ and let $P=[s=v_1, \ldots, v_{2d-1}=t]$ be their shortest path in $G[S_j]$. Recall that we assume that the diameter $D$ is even, and later on we explain the modifications for the odd diameter case. The structure of the argument is as follows. We claim that at least \emph{one} the following three scenarios hold, w.h.p., for the path $P$:
\begin{itemize}
\item{(O1)} $\dist_{H}(v_1,v_d)=O(k_D)$
\item{(O2)} $\dist_{H}(v_{d+1},v_{2d-1})=O(k_D)$
\item{(O3)} $\dist_{H}(v_1,v_{2d-1})=O(k_D)$
\end{itemize}
In the case where (O3) holds, we are done. Assume that (O1) holds. In this case, the argument is applied recursively on the path $P'=[v_{d+1},\ldots, v_{2d-1}]$. Since the depth of the recursive argument is $O(\log n)$, and each step provides a shortcut of length $O(k_D)$, the final bound on the diameter will be $O(k_D \log n)$. The same holds in a symmetric manner when (O2) holds. In order to prove that w.h.p. one of these three scenarios must hold, we introduce the notion of \emph{shortcut trees} which plays a key role in the dilation analysis. 

We need the following notations. For a node $u \in V$ and a subset $X \subseteq V$, let $\dist_G(u,X)=\min_{v \in X} \dist_G(u,v)$. For a forest $T$, let $\pi_{T}(u,v)$ be the unique $u$-$v$ path in $T$ if such exists. For $v \in (T)$, let $T(v)$ be the sub-tree rooted at $v$ in $T$.

\subsection{Shortcut Trees}\label{sec:shortcut-tree}
A shortcut tree is a spanning tree of the following auxiliary graph $G_{P,Q,\ell}$ defined for a path $P=[p_1,\ldots, p_{2d-1}]$, a node-set $Q=\{q_1,\ldots, q_{d'}\}$, as well as, an integer $\ell$ that provides an upper bound on the distance between $P$ and $Q$ in $G$, defined by $\dist_G(P,Q)=\max_{u \in P}
\dist_G(u,Q)$. The main purpose of this auxiliary graph is to fix the length of all $V(P) \times Q$ shortest paths to be exactly $\ell$. This is achieved by repeating the appearance of certain nodes on $V(P) \times Q$ paths that are shorter than $\ell$, as explained next. 

\paragraph{Auxiliary graph definition.}
The graph $G_{P,Q,\ell}$ is a layered-graph made of a collection of $\ell+2$ layers $L_1,\ldots, L_{\ell+2}$, where $L_1=V(P)$, $L_{\ell+1}=Q$ and $L_{\ell+2}=\{r\}$ (representing the root). For every $i\in \{2,\ldots, \ell\}$, each layer $L_i=\{v^i_1,\ldots, v^i_n\}$ consists of the copies of all nodes $V(G)=\{v_1,\ldots, v_n\}$ in $G$. 
The edge set of $G_{P,Q,\ell}$ consists of the following subsets of edges. The root $r$ is connected to every node in $L_{\ell+1}=Q$. For consistency with the other layers, let $L_1=V(P)=[p_1,\ldots, p_{2d-1}]$ be also referred to by $[v^1_1,\ldots, v^1_{2d-1}]$ (i.e., $p_j=v^1_j$ for every $j \in \{1,\ldots, 2d-1\}$).
In addition, for every $i \in \{1,\ldots, \ell\}$, the edges between layers $L_{i}$ and $L_{i+1}$ correspond to the $G$-edges between these nodes. 
Specifically, define:
$$E(L_i,L_{i+1})=\{ (v^i_j, v^{i+1}_j) ~\mid~ v^i_j \in L_i, v^{i+1}_j \in L_{i+1}\} \cup \{(v^i_j, v^{i+1}_k) ~\mid~ v^i_j \in L_i, v^{i+1}_k \in L_{i+1} \mbox{~and~}(v_j,v_k) \in E(G)\}~.$$
%
%
That is, each node $v^i_j$ in layer $i$ is connected in layer $i+1$ to its own copy $v^{i+1}_j$, as well as to the $(i+1)^{th}$ copies of its neighbors in $G$. 
The edge set of $G_{P,Q,\ell}$ is then given by
$$E(G_{P,Q,\ell})=\{(r,q_j) ~\mid~ q_j \in Q\} \cup \bigcup_{i=1}^{\ell}E(L_i,L_{i+1})~.$$
Note that all edges, except those incident to $r$ and the self-copies edges, correspond to $G$-edges. In addition, each $G$-edge $(u,v)$ (in this direction) has at most one corresponding edge in each subset $E(L_i,L_{i+1})$ for $i \in \{1,\ldots, \ell\}$. 

\paragraph{Auxiliary tree definition.} Let $T_{P,Q,\ell}$ be a BFS tree rooted at $r$ in the graph $G_{P,Q,\ell}$. The BFS tree $T_{P,Q,\ell}$ has depth $\ell+2$, and since $\dist_G(P,Q)\leq \ell$, it holds that each leaf node $p_i \in P$ is connected to the root $r$. (That is, the leaf set of $T_{P,Q,\ell}$ is precisely $V(P)$). Also note that $T_{P,Q,\ell}$ might contain only a strict subset of the nodes of $G_{P,Q,\ell}$, see Fig. \ref{fig:aux-graph} (Left) for an illustration. We next describe a random sparsification of the tree $T_{P,Q,\ell}$ which mimics the edge sampling into the shortcut subgraph $H$ in Step (2) of the centralized construction.  Define the graph $T_{P,Q,\ell}[\mathbf{p}]\subseteq T_{P,Q,\ell}$ by sampling each \emph{non-self edge}\footnote{An edge in $E(L_i,L_{i+1})$ is non-self if it connects copies of distinct vertices in $G$.} in 
$$\bigcup_{i=2}^{\ell}E(L_i,L_{i+1}) \cap E(T_{P,Q,\ell})$$
independently with probability $\mathbf{p}=\log n\cdot k_D/N=\log n \cdot n^{-1/(D-1)}$. The edges in $E(L_1,L_2)$ and the edges incident to $r$ are taken into $T_{P,Q,\ell}[\mathbf{p}]$ with probability of $1$. (In our argument, $P$ will correspond to some shortest path in $G[S_j]$, and since in Step (1) of the algorithm, all edges incident to $S_j$ are taken into $H$ with probability $1$, the corresponding edges in the tree $T_{P,Q,\ell}$ are also included in $T_{P,Q,\ell}[\mathbf{p}]$.) 
For our purposes in the construction of $T_{P,Q,\ell}[\mathbf{p}]$, we use the same randomness used in the construction of the shortcut subgraph $H$. Recall that in Step (2) of the shortcut algorithm, each (directed) edge in $(v_i,v_{i'}) \in E$ with $v_i \notin S_j$ is sampled by $v_i$, $D$ independent times each with probability $\mathbf{p}$. The edge $(v^k_i, v^{k+1}_{i'}) \in E(L_k,L_{k+1}) \cap E(T_{P,Q,\ell})$ for $k \in \{2,\ldots, \ell\}$ is taken into $T^*$ based on the $(k-1)^{th}$ sampling of the edge $(v_i,v_{i'})$ in Step (2). 

We now describe the edges in $T_{P,Q,\ell}[\mathbf{p}]$ formally. Recall that $E_k$ are all $G$-edges sampled into $H$ in the $k^{th}$ application of Step (2) for every $k \in \{1,\ldots, D\}$. Then, formally the graph $T_{P,Q,\ell}[\mathbf{p}]$ consists of the following edges:
\begin{itemize} 
\item $(E(L_1,L_2) \cup E(L_{\ell+1}, L_{\ell+2})) \cap T_{P,Q,\ell}$,
\item self-edges: $\{(v_i^k, v_i^{k+1})\in E(L_k,L_{k+1}) \cap E(T_{P,Q,\ell}) ~\mid~ v_i \in L_k, k \in \{2,\ldots, \ell\}\}$,

\item sampled non self-edges: $\{(v_i^k, v_j^{k+1}) \in E(L_k,L_{k+1}) \cap E(T_{P,Q,\ell}) ~\mid~ (v_i,v_j) \in E_{k-1}, k \in \{2,\ldots, \ell\}\}$. 
\end{itemize} 
That is, for every $k \in \{2,\ldots, \ell\}$, each edge $(v_i^k, v_j^{k+1}) \in E(L_k,L_{k+1}) \cap E(T_{P,Q,\ell})$ is added only if it was sampled in the $(k-1)^{th}$ repetition of Step (2). 
%
Finally, the dilation argument is applied on the graph $T^*_{P,Q,\ell}=T_{P,Q,\ell}[\mathbf{p}] \cup E(P)$.
See Fig. \ref{fig:aux-graph} for an illustration. 

\begin{figure}[h!]
\includegraphics[scale=0.28]{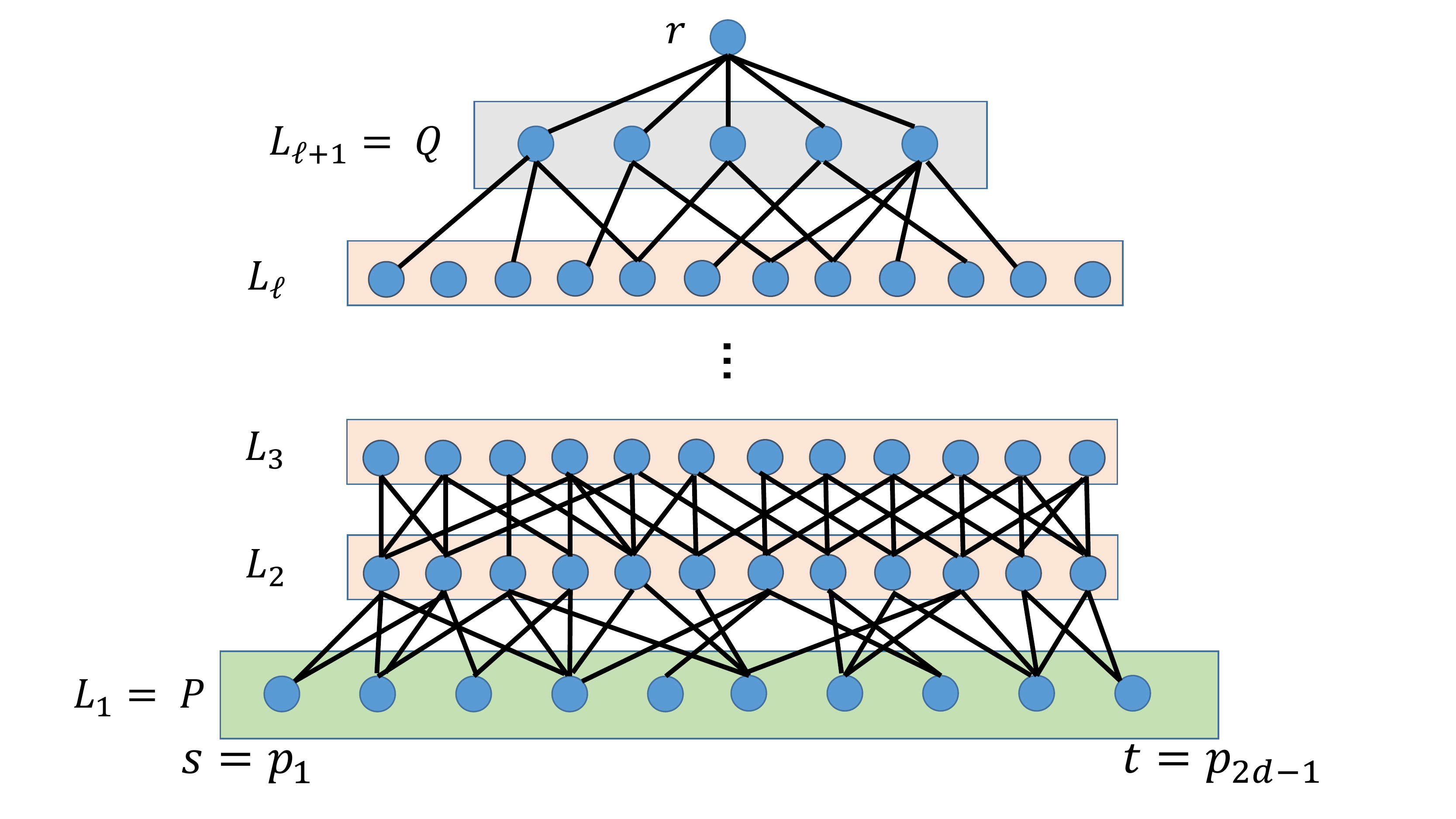}
\includegraphics[scale=0.28]{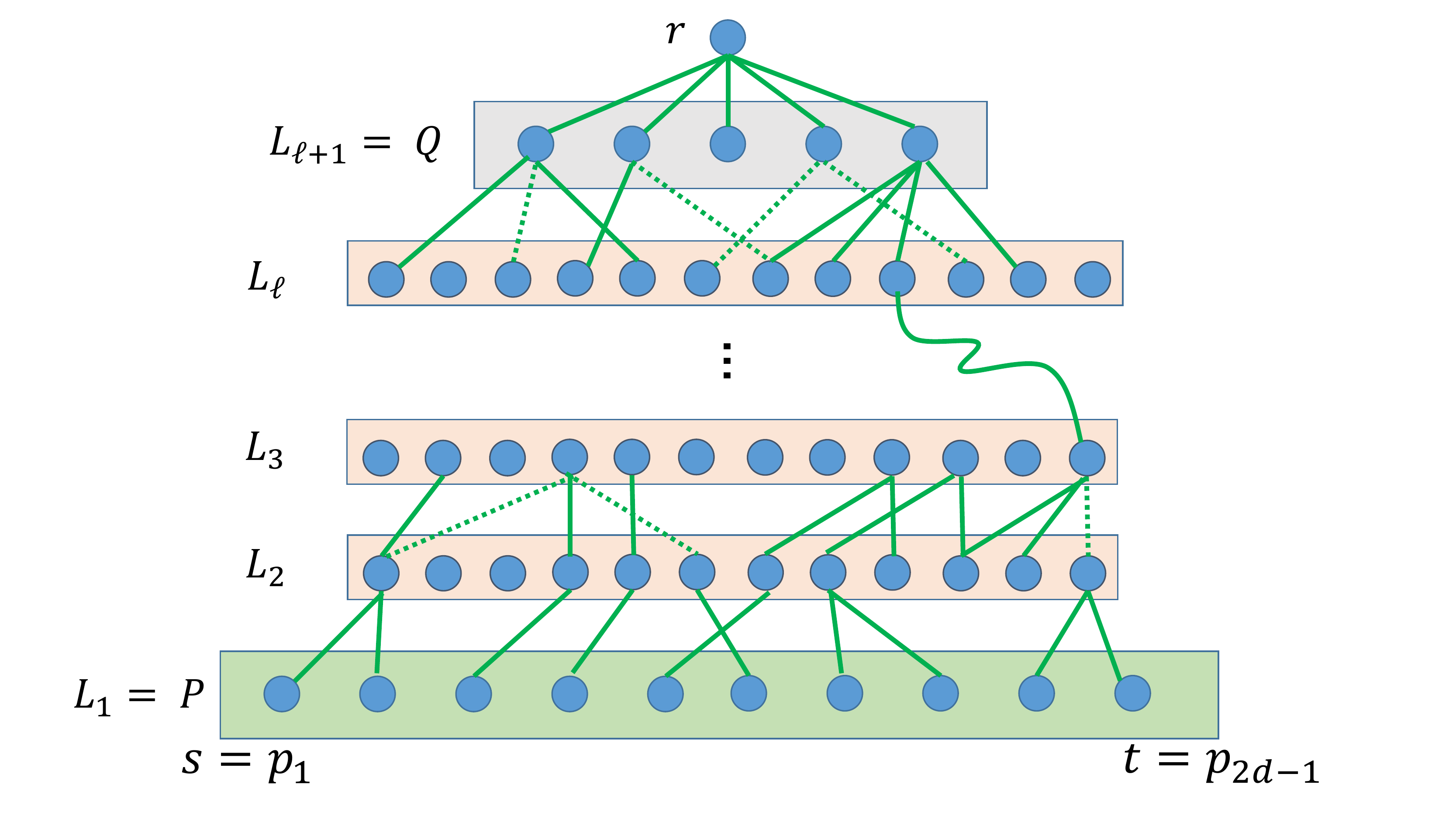}
\caption{\sf Left: An illustration of the auxiliary graph $G_{P,Q,\ell}$ and the BFS tree $T_{P,Q,\ell}$ shown in green. Each node in $P$ is connected by an $(\ell+1)$-length path to the root in $T_{P,Q,\ell}$. Right: The sampled graph $T^*$ is obtained by sampling each of the BFS edges between layers $L_i$ and $L_{i+1}$ independently with probability of $\mathbf{p}=n^{-1/(D-1)}$, for every $i \in \{1, \ldots, \ell\}$. The sampled edges of $T^*$ are shown in dashed. \label{fig:aux-graph} 
}
\end{figure}

To avoid cumbersome notation, let $T^*=T^*_{P,Q,\ell}$. The key lemma in our context is the following. Let $s=p_1$ and $t=p_{2d-1}$.
\BL\label{lem:key-tstar}
W.h.p., either $\dist_{T^*}(s,t)=O(k_D)$, or else it must hold (w.h.p.) that
$\dist_{T^*}(s,L_j)=O(k_D)$ for every $j \in \{2,\ldots, \min\{\ell+1, D/2+1\}\}$. The high probability guarantee 
\EL
To prove Lemma \ref{lem:key-tstar} we introduce the notion of $(i,k)$ walks. Throughout, the randomized arguments are applied over the sampling of the edges of $T_{P,Q,\ell}$ into $T^*$. 

\paragraph{Probabilistic analysis of $(i,k)$ walks.} An $(i,k)$ walk in the graph $T^*$ is a walk that starts at a node $p_i$ (the $i^{th}$ node of the path $P$, in layer $1$) and ends at some node in the set $V^+_i(P) \cup L_k$ where $V^+_i(P)=[p_i,\ldots, p_{2d-1}]$. To describe the structure of a legal $(i,k)$ walk, it is convenient to view the nodes in layer $L_1$ of $T^*$ ordered from left to right, namely, $L_1=V(P)=(p_1,\ldots, p_{2d-1})$. 

\BD[$(i,k)$ unit]
An $(i,k)$ \emph{unit} is a walk that starts at node $p_{i}$ and ends at node $p_j$ for $j\geq i$ defined as follows. 
Let $u_{i,k}$ be the up-most ancestor of $p_i$ in $T^* \cap \bigcup_{\ell=2}^{k} L_\ell$. That is, the ancestor of $p_i$ in the maximum layer $\ell \leq k$. 
Letting $p_{j}$ be the right-most $P$-node (i.e., node of largest $j$ index on $P$) in the subtree of $T^*(u_{i,k})$, then the $(i,k)$ unit is defined by the $T^*$-walk:
$$P'=\pi_{T^*}(p_i, u_{i,k}) \circ \pi_{T^*}(u_{i,k},p_{j})~.$$ 
\ED
Note that $(i,k)$ unit is indeed a walk (and not a simple path), in the case where the least-common-ancestor of $p_i$ and $p_j$ in $T^*$ is strictly below (i.e., in a smaller level then) $u_{i,k}$.  A \emph{maximal} $(i,k)$ walk is a walk $P''$ defined by applying $2d-i$ steps defined as follows. Initially, let $P_1$ be an $(i,k)$ unit, and let $p_{i_1}$ be the second endpoint of $P_1$, thus $i_1\geq i$. At any step $j \in \{1,\ldots, 2d-i\}$, let $P_j$ be a walk ending at some node $p_{i_j}$. If $p_{i_j}=t$, then let $P_{j+1}=P_j$. Otherwise, let $P'_j$ be an $(i_j+1,k)$-unit, and define $P_{j+1}=P_j \circ (p_{i_j}, p_{i_j+1}) \circ P'_j$. The maximal $(i,k)$ walk is given by $P''=P_{2d-i}$. An $(i,k)$-\emph{walk} is \emph{any} sub-walk of a maximal $(i,k)$ walk. See Fig. \ref{fig:ikpath} for an illustration.
%
%

\begin{figure}[h!]
\begin{center}
\includegraphics[scale=0.40]{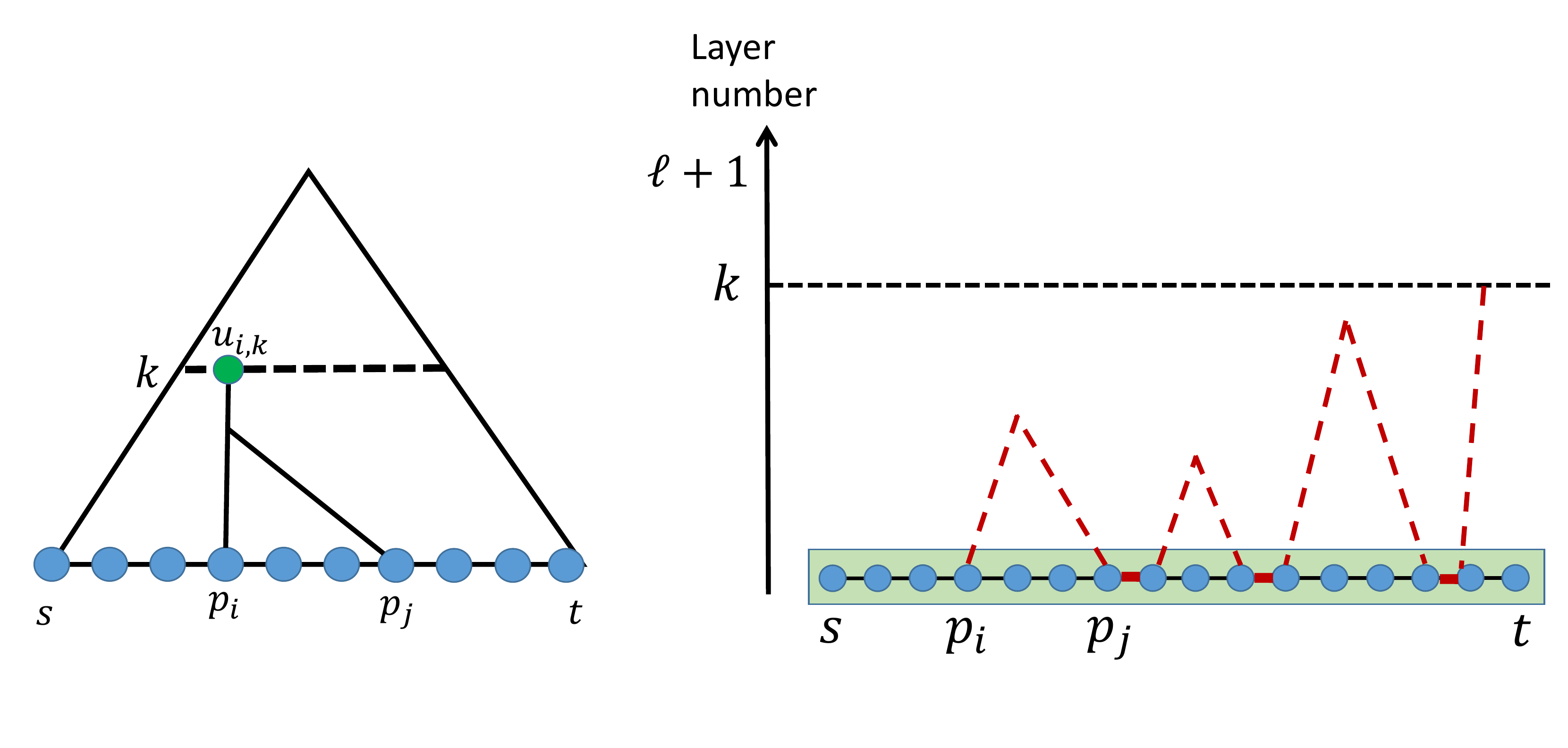}
\caption{\sf Left: The forest $T^*$ and an $(i,k)$ unit $P'=\pi_{T^*}(p_i,u_{i,k})\circ \pi_{T*}(u_{i,k},p_j)$, where $j\geq i$. Since the least-common-ancestor of $p_i,p_j$ in $T^*$ is below $u_{i,k}$, $P'$ is indeed a walk (and not a simple path). Note that $P'$ contains a unique $k$-level node, namely, $u_{i,k}$. Right: an $(i,k)$ walk made by concatenating $(j,k)$ walks, for increasing $j$ values, interleaved with edges from $P$. \label{fig:ikpath} 
}
\end{center}
\end{figure}

Let $P'=[p_i=u_1,\ldots, u_b]$ be an $(i,k)$ walk, and consider the subset of $P$-nodes in $P'$ ordered according to their appearance in $P'$. We first observe that the sequence of nodes in $P'$ is \emph{monotone} in the sense of being sorted from left to right. Also observe that $P'$ can be written as a concatenation of walks $P''=[u,v]$ where $u \in P$ and $v \in P \cup L_k$.

\begin{obs}\label{path-distinct}
Let $w_1,\ldots, w_a$ be the multi-set of level $k$ nodes (i.e., in $L_k$) of an $(i,k)$ walk $P'$ sorted according to their appearance on $P'$ (in increasing distance from $p_i$). Then, $w_j \neq w_{j'}$ for every $j \neq j' \in \{1,\ldots, a\}$.
\end{obs}
\begin{proof}
Let $P_1,\ldots, P_a$ be the $(i_j,k)$ units composing $P'$ such that $w_j \in P_j$ for every $j \in \{1,\ldots, a\}$. Note that each $P_j$ contains a unique level $k$ node. It then holds that $i< i_1 < i_2 \ldots < i_a$. Assume towards contradiction that $w_j=w_{j'}$ where $j< j'$. 
By the definition of $P_{j'}$, we have that $p_{i_{j'}}$ is in the sub-tree of $w_{j}=w_{j'}$, in contradiction to the definition of $P_{j}$ which ends in a node $p_{i_g}$ for $i_g < i_{j'}$. 
The claim follows.
\end{proof}

The key lemma in our context shows that w.h.p. the graph $T^*$ contains short $(i,k)$ walks for every $i$ and $k$.
\BL\label{lem:length-ik}
For every $i \in \{1,\ldots, 2d-1\}$ and $k \in \{2,\ldots, \ell+1\}$, $T^*$ contains an $(i,k)$ walk between $p_i$ to a node in $\{t\} \cup L_k$ of length at most $(c \cdot k_D/N)^{-k+2}$ for a sufficiently large constant \footnote{This constant effects the high probability guarantee of $1-1/n^{c'}$ on the final diameter bound, where $c'$ is some constant that depends on $c$.} $c \geq 8$, w.h.p. Moreover, this high probability guarantee uses at most $k$ out of $D$ (independent) repetitions of Step (2). 
\EL
\begin{proof}
The proof is shown by induction on $k$. The base case of $k=2$ follows for every $i \in \{1,\ldots, 2d-1\}$, as all edges of $E(L_1,L_2)\cap T_{P,Q,\ell}$ are kept in $T^*$ with probability of $1$. Let $R_{i,k}$ be an indicator random variable for the event that there exists an $(i,k)$ walk of length at most $\ell_k=(c \cdot k_D/N)^{-k+2}$ ending at some node in $L_k \cup \{t\}$. 


Assume that w.h.p. $R_{i,k}=1$ for every $i \in \{1,\ldots, 2d-1\}$. We next show that also $R_{i,k+1}$ holds w.h.p. for every $i$. Let $R_{k}$ be the indicator random variable that $R_{i,k}=1$ for every $i \in \{1,\ldots, |P|\}$. Thus, by induction hypothesis, $R_k=1$ w.h.p. as well. Next, observe that
$$Pr[R_{i,k+1}] \geq Pr[R_{i,k+1}~\mid~ R_k] \cdot Pr[R_k]~.$$
As $Pr[R_k] \geq 1-1/n^{c''}$ for some constant $c''\geq 2$, it is sufficient to bound the probability $Pr[R_{i,k+1} ~\mid~ R_k]$. 
To do that, we fix an $i \in \{1,\ldots, 2d-1\}$, and bound the probability of having a particular $(i,k+1)$ walk. 
We define a walk $P'$ that starts at $p_i$ and ends at a $P$-node. The walk $P'$ is defined in $\ell=(c/2) \cdot N/k_D$ steps. Let $P_1$ be an $(i,k)$ unit. In step $j \in \{1,\ldots, \ell-1\}$, we are given a path $P_1 \circ \ldots \circ P_j$ that ends in $p_{d_j}$. If $p_{d_j}\neq t$, let $P_{j+1}=(p_{d_j},p_{d_j+1}) \circ P'_{j+1}$, where $P'_{j+1}$ is an $(d_j+1,k)$ unit. Otherwise (if $p_{d_j}=t$), $P_{j+1}=\emptyset$. 

Let $P'=P_1 \circ P_2 \circ \ldots \circ P_\ell$. If $P'$ ends in $t$, we are done since $|P'|\leq \ell (\ell_k +1)\leq \ell_{k+1}$. It remains to consider the case where none of the $P_j$ paths ends in $t$. 
For every path $P_j$ for $j \in \{1,\ldots, \ell\}$, let $w_j$ be the unique level-$k$ node in $L_k \cap V(P_j)$. 
Since $P'$ is an $(i,k)$ walk, by Obs. \ref{path-distinct}, $w_j \neq w_{j'}$ for every $j \neq j' \in \{1,\ldots, \ell\}$.
We now show how to use $P'$ to obtain an $(i,k+1)$ walk of length at most $\ell_{k+1}$. To see this, let $e_j=(w_j, par(w_j))$ be the edge connecting $w_j$ to its parent in the tree $T_{P,Q,\ell}$. By the observation, $e_j \neq e_{j'}$ for every $j \neq j' \in \{1,\ldots, \ell\}$. As each edge $e_j$ is sampled independently into with probability\footnote{This is because $e_j$ is taken into $T^*$ only if its corresponding (directed) $G$ edge was sampled into $E_{k-1}$.} $\mathbf{p}$ into $T^*$, the probability that at least one of these edges, say $e_j$, is in $T^*$ is at least
$$1-(1-\mathbf{p})^\ell =1-(1-\log n \cdot k_D/N)^{(c/2) \cdot N/k_D}\geq 1-1/n^{c/3}~.$$

This yields the $(i,k+1)$ walk $P''=P'[p_i,w_j] \circ (w_j, par(w_j))$ of length at most $\ell\cdot (\ell_k+2)\leq \ell_{k+1}$ as desired. Note that the bound on $(i,k)$ walks only exploits the randomness in the sampled sets $E_1,\ldots, E_{k}$, i.e., the first $k$ sampling repetitions of Step (2). 
\end{proof}
Lemma \ref{lem:key-tstar} then follows by applying Lemma \ref{lem:length-ik} with $i=1$ (as $p_1=s$) and $k\leq D/2+1$. Note that for $k=D/2+1$, the length of the $(1,k)$ walk is bounded by $(c \cdot k_D/N)^{-D/2+1}=O(k_D)$ as desired.  Next, we show the equivalence between an $(i,k)$ walk in $T^*$ to a corresponding path in the subgraph $H$ (for which we provide the diameter bound). 

\begin{obs}\label{lem:relations}
Let $P'$ be an $(i,k)$ walk of length $\ell'$ in $T^*$ with endpoints $p_i$ and $u \in V(P) \cup L_k$. Then, there exists a path in $H$ between $p_i$ and the $G$-copy of $u$ of length at most $\ell'$. 
\end{obs}
By Lemma \ref{lem:key-tstar} and Obs. \ref{lem:relations}, we have:
\BCR\label{cor:final-arg}
Let $P,Q$ be such that $\dist_G(P,Q)\leq \ell$. Then, w.h.p, either $\dist_{H}(s,t)=O(k_D)$ or else, for every $k \in \{2,\ldots, \min\{\ell+1, D/2+1\}\}$, there exists a node $u \in V(G)$ such that $\dist_G(u, Q)\leq \ell-k+1$ and $\dist_{H}(s,u)=O(k_D)$. Moreover, the  probabilistic argument uses at most $k$ repetitions out of the $D$ repetitions of Step (2) of the centralized construction.
\ECR


\subsection{Proof of Theorem \ref{thm:dilation}}\label{sec:comp-arg}
Equipped with the tool of shortcut trees and $(i,k)$ walks, we are now ready to provide the dilation argument of the subgraph $H=G[S_j] \cup H_j$. Our goal is to show that $\dist_H(s,t)=\widetilde{O}(k_D)$ for a fixed pair $s,t$ in $S_j$. Let $P'=[s=v_1,\ldots, v_{2d-1}=t]$ be the $s$-$t$ shortest path in $G[S_j]$.
The next key lemma shows that at least half of the path $P'$ can be shortcut in $H$ into a path length of $\widetilde{O}(k_D)$. Since this argument can be applied to any sub-path of $P'$, the final dilation bound is obtained by a recursive application of that lemma. We show:

\BL\label{lem:start-half}
W.h.p., one of the three events must hold w.h.p (i) $\dist_{H}(v_{d+1},v_{2d-1})=O(k_D)$, (ii) $\dist_{H}(v_{1},v_{d})=O(k_D)$, or (iii) $\dist_{H}(v_{1},v_{2d-1})=O(k_D)$. 
\EL
\begin{proof}
The proof is based on having $d+1$ applications of the shortcut trees of Section \ref{sec:shortcut-tree}. See Fig. \ref{fig:dilation-arg} for an illustration. Let $H_1, H_2$ be the edges added to the subgraph in the first (resp., last) $D/2$ applications of Step (2) of the algorithm. We will show that w.h.p. over the randomness of the edges sampled to $H_1$, each of the $d$ applications of the shortcut trees satisfies a certain desired property. Then, conditioned on these properties, we show that the last $(d+1)^{th}$ application satisfies another property w.h.p. over the randomness of the edges sampled to $H_2$. Since there is a dependency here, and each application uses at most $D/2$ sampling steps of Step (2), over all we need $D$ sampling steps. (For the first $d$ applications, we use the same $D/2$ sampling steps, as there is no conditioning between these applications).

Let $P_1=[s=v_1,\ldots, v_d]$ be the first half of the path $P'$, and let $P_2=[t=v_{2d-1}, \ldots, v_{d+1}]$ be the second half of the path, written in a \emph{reverse} manner from $v_{2d-1}$ to $v_{d+1}$. 
We start by making $d$ applications of the shortcut tree constructions where for each $i \in \{1,\ldots, d\}$, we define $Q_i=\{v_i\}$ and the auxiliary graph $G_{P_2,Q_i,D}$, the tree $T_{P_2,Q_i,D}$ and the final graph $T^*_i=T_{P_2,Q_i,D}[\mathbf{p}] \cup E(P_2)$.

By Cor. \ref{cor:final-arg}, it holds that w.h.p. one of the following two events hold for any $i \in \{1,\ldots, d\}$ using at the most $D/2$ repetitions of the edge sampling in Step (2) of the algorithm.
\begin{itemize}
\item{(E1)} There exists a node $v_{i_j} \in V(H)$ such that $\dist_{H}(v_{i_j},v_{2d-1})=(c N/k_D) ^{(D-2)/2}=O(k_D)$ and $v^{\ell'}_{i_j}$ for $\ell'=D/2+1$ is a node in layer $\ell'$ in $T^*_i$. 
\item{(E2)} $\dist_{H}(v_{d+1},v_{2d-1})=O(k_D)$.
\end{itemize}
%
If there exists an index $i \in \{1,\ldots, d\}$ for which event (E2) holds, we are done. Assume from now on that the event (E1) holds w.h.p. for every $i \in \{1,\ldots, d\}$. Let $q_i=v_{i_j}$ and define $Q=\{q_1,\ldots, q_d\}$. Note that the nodes $q_i$ are not necessarily distinct. Since each $q_i$ appears in level $D/2+1$ of the tree $T_{P_2,\{v_i\},D}$, it holds that $\dist_G(v_i,q_i)\leq D/2$. Therefore $\dist(V(P_1), Q)\leq D/2$. We now apply again the shortcut tree construction and define the graphs 
$G_{P_1,Q,D/2}$, $T_{P_1,Q,D/2}$ (as in Sec. \ref{sec:shortcut-tree}). The sampled sub-tree $T_{P_1,Q,D/2}[\mathbf{p}]$ is based on the sampled edge set $E_{D/2+1}, \ldots, E_D$ (i.e., the edges sampled in the last $D/2$ applications of Step (2)). Formally, letting $L_1=V(P_1)$, $L_{D/2+1}=Q$ and $L_{D/2+2}=\{r\}$, then $T_{P_1,Q,D/2}[\mathbf{p}]$ consists of the following edges:
\begin{itemize} 
\item $(E(L_1,L_2) \cup E(L_{D/2+1}, L_{D/2+2})) \cap T_{P_1,Q,D/2}$,
\item self-edges: $\{(v_i^k, v_i^{k+1})\in E(L_k,L_{k+1}) \cap E(T_{P_1,Q,D/2}) ~\mid~ v_i \in L_k, k \in \{2,\ldots,D/2\}\}$,

\item sampled non self-edges: $\{(v_i^k, v_j^{k+1}) \in E(L_k,L_{k+1}) \cap E(T_{P_1,Q,D/2}) ~\mid~ (v_i,v_j) \in E_{D/2+k-1}, k \in \{2,\ldots, D/2\}\}$. 
\end{itemize}  

Let $T^*_1=T_{P_1,Q,D/2}[\mathbf{p}] \cup E(P_1)$. By Cor. \ref{cor:final-arg} it holds that w.h.p. (over the second set of $D/2$ repetitions of the edge-sampling in Step (2)) that one of the following two events holds:
\begin{itemize}
\item{(E3)} $\exists q_{i^*} \in V(H) \cap Q$ such that $\dist_{H}(v_1,q_{i^*})=(c N/k_D)^{(D-2)/2}=O(k_D)$.
\item{(E4)} $\dist_{H}(v_{1},v_{d})=O(k_D)$.
\end{itemize}
If event (E4) holds we are done. Thus, consider the case where (E3) holds. By event (E2), we have the w.h.p. $\dist_{H}(q_i, v_{2d-1})=O(k_D)$ for every $i \in \{1,\ldots, d\}$. By combining with (E3) we have that 
$\dist_H(v_1,v_{2d-1})\leq \dist_H(v_1,q_{i^*})+ \dist_H(q_{i^*},v_{2d-1})=O(k_D)$ as required. The lemma follows.
\end{proof}


\begin{figure}[h!]
\begin{center}
\includegraphics[scale=0.40]{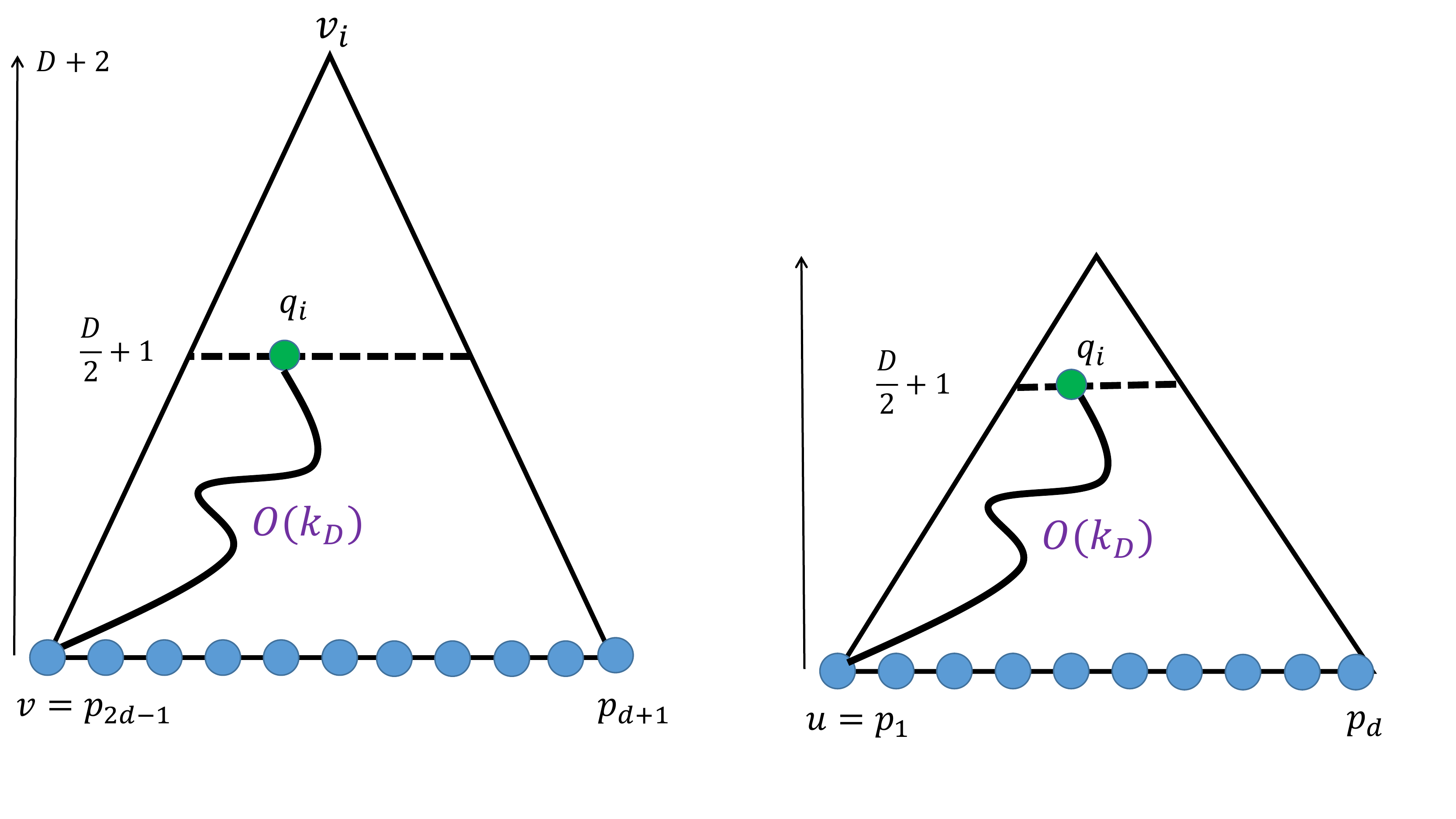}
\caption{\sf Illustration for the dilation argument of Theorem \ref{thm:dilation} for an $u$-$v$ shortest path $P=[s=v_1,\ldots, t=v_{2d-1}]$. The argument applies $d+1$ applications of the shortcut trees scheme. Left: An illustration for the auxiliary graph $T^*_i=T_{P_2,Q_i,D}[\mathbf{p}] \cup E(P_2)$. W.h.p., $T^*_i$ contains either (i) an $v_{2d-1}$-$v_d$ path of length $O(k_D)$ or else (ii) an $v_{2d-1}$-$q_i$ path of length $O(k_D)$ for some node $q_i$ in level $D/2+1$ of the tree $T_{P_2,Q_i,D}$. Right: The second part of the argument applies the shortcut construction for the auxiliary graph $T^*_1=T_{P_1,Q,D}[\mathbf{p}] \cup E(P_1)$ where $Q$ is the collection of $q_i$ nodes defined by the prior $d$ applications, for each $v_i \in P_1$. W.h.p., it then holds that $T^*_1$ contains either  (i) an $v_{1}$-$v_d$ path of length $O(k_D)$ or else (ii) an $v_{1}$-$q_i$ path of length $O(k_D)$ for some node $q_i$. In the latter case, we obtain an $v_1$-$v_{2d-1}$ shortcut path of length $O(k_D)$ that goes through $q_i$.  \label{fig:dilation-arg} 
}
\end{center}
\end{figure}

\begin{proof}[Proof of Theorem \ref{thm:dilation}]
Consider $s,t \in G[S_j]$ and let $P=[s=v_1,\ldots, v_{2d-1}=t]$ be an $s$-$t$ shortest path in $G[S_j]$. By Lemma \ref{lem:start-half}, it holds that it least half of the $s$-$t$ shortest path $P \subseteq G[S_j]$ can be shorten by a path of length $K=O(k_D)$. In the same manner, we apply Lemma \ref{lem:start-half} on any subpath $P' \subseteq P$. Since there are at most $|P|^2$ such paths, by the union bound, the guarantee of Lemma \ref{lem:start-half} holds for any such sub-path $P' \subseteq P$. 

Coming back to our path $P$, by Lemma \ref{lem:start-half} the shortcut argument can be applied recursively on the remaining path $P'$ where $|P'|\leq |P|/2$. Since in each recursive application w.h.p. there is a shortcut in one of the two bisections of the path, overall we obtain an $s$-$t$ path in $H$ of length $(k_D \log n)$ w.h.p.

To handle the case where the diameter is odd, the algorithm is modified as follows. We split each edge $e=(u,v)$ in $G$ into \emph{two edges} by introducing a dummy intermediate node $x_e$ connected (only) to $u$ and $v$. The resulting modified graph $G'$ has now an even diameter $D'=2D$. Note that any path in $G'$ corresponds to a path in $G$ in the following manner: 
The shortcut algorithm is applied on the graph $G'$ where the only modification is that the sampling probability of each edge in $G'$ is set to $\mathbf{p}'=\sqrt{\mathbf{p}}$ (except for the edges chosen in Step $1$ of the algorithm, for each such edge we take the  two-length corresponding path with probability $1$). 
Note that two edges $(u,x_e)$ and $(x_e,v)$ are sampled into the shortcut subgraph $H'_j$ with probability of $(\mathbf{p}')^2=\mathbf{p}$. The final output subgraph $H_j$ contains only edges $(u,v)$ such that both $(u,x_e)$ and $(x_e,v)$ are sampled into $H'_j$. We then apply the argument similarly to the even diameter case. More specifically, we will be working on the graph $G'$ all along. Then, in Lemma \ref{lem:start-half} the set $Q$ defined based on the first $d$ applications of the shortcut tree argument correspond to dummy nodes. Then, the $(d+1)^{th}$ application shows that either there is a short $G'$ path of length $O(k_D)$ between the endpoints $s,t$ of the path, or a shortcut of length $O(k_D)$ between a path endpoint to a mid-point on the path. The reason that the length remains $O(k_D)$ in this construction is that a path from level $1$ to level $D+1$ in the shortcut tree of $G'$ 
contains $D-2$ edges where each edge was chosen with probability $\sqrt{\mathbf{p}}$ (we do not take into account the first two edges as they were chosen with probability $1$ into the shortcut subgraph), and thus the path length is $O(k_D)$ by the application of Lemma 
\ref{lem:length-ik} and Lemma \ref{lem:start-half}. The rest of the argument is almost identical to the even case and thus omitted.
\end{proof}

\section{Applications to Distributed Optimization}\label{sec:app}

\BF[\cite{GhaffariThesis17}]\label{fc:shortcut-MST}
Let $\mathcal{G}$ be a graph family such that for each graph $G \in \mathcal{G}$ and
any partition of $G$ into vertex-disjoint connected graphs $G_1,\ldots,G_N$, one can find an 
$(\congestion,\dilation)$ shortcuts such that $\congestion+\dilation\leq K$ and these shortcuts can be computed in $\widetilde{O}(K)$ rounds.
Then: 
\begin{itemize}
\item{[Theorem 6.1.2]:} there is a randomized distributed MST algorithm that computes an MST in $\widetilde{O}(K)$ rounds, with high probability, in any graph from the family $G$. 
\item{[Theorem 7.6.1]:}  there is a randomized distributed algorithm that computes a $(1+\epsilon)$ approximation of the minimum cut in $\widetilde{O}(K)$ rounds, with high probability, in any graph from the family $G$. 
\end{itemize}
\EF
Corollary \ref{cor:MST} follows by combining Fact \ref{fc:shortcut-MST} with Theorem \ref{thm:main-result}. 
An additional immediate corollary of improved shortcuts is for computing an approximate SSSP. 
Haeupler and Li \cite{HaeuplerL18} provided improved algorithms for several shortest-path problems whose bounds depend on the quality of shortcuts. By plugging the bounds of Theorem \ref{thm:main-result} into Corollaries 2,3 in \cite{HaeuplerL18} we get:
\BCR[Improved Distributed SSSP Tree Algorithms]\label{cor:imp-sssp}
There are randomized algorithms, that, given an $n$-vertex weighted graph with polynomial edge weights of unweighted diameter $D=O(1)$ perform the following tasks:
(1) compute a spanning tree that approximates distances to a given source vertex to within factor $(\log n)^{O(1/\epsilon)})$, in $\widetilde{O}(n^{(D-2)/(2D-2)})\cdot n^{\epsilon})$ rounds for any constant $\epsilon$; and  (2) compute  a spanning tree that approximates distances to a given source vertex within factor $2^{O(\sqrt{\log n})}$,
in $\widetilde{O}(n^{(D-2)/(2D-2)})\cdot 2^{O(\sqrt{\log n})})$ rounds.
\ECR
Finally, Dory and Ghaffari \cite{dory2019improved} recently studied the distributed approximation of minimum weight two-edge connected subgraphs ($2$-EECS). By plugging Theorem \ref{thm:main-result} into Theorem 1.2 of \cite{dory2019improved}, we get:
\BCR[Improved Approximation of $2$-EECS]\label{cor:twoeecs}
There is an algorithm, that, given an $n$-vertex weighted graph of (unweighted) diameter $D=O(1)$, computes an $O(\log n)$-approximation of the weighted 2-ECSS in $\widetilde{O}(n^{(D-2)/(2D-2)})$ rounds, with high probability.
\ECR

\paragraph{Acknowledgments.} We are very grateful to the PODC 2021 reviewers for many insightful comments, and specifically for the reviewer suggesting the improvement of Lemma \ref{lem:length-ik}. 

\newpage
\bibliographystyle{alpha}
\bibliography{ref}

\newcommand{\etalchar}[1]{$^{#1}$}
\begin{thebibliography}{LPPSP03}

\bibitem[AJB99]{albert1999diameter}
R{\'e}ka Albert, Hawoong Jeong, and Albert-L{\'a}szl{\'o} Barab{\'a}si.
\newblock Diameter of the world-wide web.
\newblock {\em nature}, 401(6749):130--131, 1999.

\bibitem[CPT20]{ChuzhoyPT20}
Julia Chuzhoy, Merav Parter, and Zihan Tan.
\newblock On packing low-diameter spanning trees.
\newblock In {\em 47th International Colloquium on Automata, Languages, and
  Programming, {ICALP} 2020, July 8-11, 2020, Saarbr{\"{u}}cken, Germany
  (Virtual Conference)}, pages 33:1--33:18, 2020.

\bibitem[DG19]{dory2019improved}
Michal Dory and Mohsen Ghaffari.
\newblock Improved distributed approximations for minimum-weight
  two-edge-connected spanning subgraph.
\newblock In {\em Proceedings of the 2019 ACM Symposium on Principles of
  Distributed Computing}, pages 521--530, 2019.

\bibitem[Elk04]{Elkin04}
Michael Elkin.
\newblock Unconditional lower bounds on the time-approximation tradeoffs for
  the distributed minimum spanning tree problem.
\newblock In {\em Proceedings of the 36th Annual {ACM} Symposium on Theory of
  Computing, Chicago, IL, USA, June 13-16, 2004}, pages 331--340, 2004.

\bibitem[GH16]{GhaffariH16}
Mohsen Ghaffari and Bernhard Haeupler.
\newblock Distributed algorithms for planar networks {II:} low-congestion
  shortcuts, mst, and min-cut.
\newblock In Robert Krauthgamer, editor, {\em Proceedings of the Twenty-Seventh
  Annual {ACM-SIAM} Symposium on Discrete Algorithms, {SODA} 2016, Arlington,
  VA, USA, January 10-12, 2016}, pages 202--219. {SIAM}, 2016.

\bibitem[GH20]{GhaffariHMinorArxiv}
Mohsen Ghaffari and Bernhard Haeupler.
\newblock Low-congestion shortcuts for graphs excluding dense minors.
\newblock {\em CoRR}, abs/2008.03091, 2020.

\bibitem[Gha15]{ghaffari2015near}
Mohsen Ghaffari.
\newblock Near-optimal scheduling of distributed algorithms.
\newblock In {\em Proceedings of the 2015 ACM Symposium on Principles of
  Distributed Computing}, pages 3--12. ACM, 2015.

\bibitem[Gha17]{GhaffariThesis17}
Mohsen Ghaffari.
\newblock {\em Improved Distributed Algorithms for Fundamental Graph Problems}.
\newblock PhD thesis, MIT, {USA}, 2017.

\bibitem[GKS17]{GhaffariKS17}
Mohsen Ghaffari, Fabian Kuhn, and Hsin{-}Hao Su.
\newblock Distributed {MST} and routing in almost mixing time.
\newblock In {\em Proceedings of the {ACM} Symposium on Principles of
  Distributed Computing, {PODC} 2017, Washington, DC, USA, July 25-27, 2017},
  pages 131--140, 2017.

\bibitem[GP16]{GhaffariP16}
Mohsen Ghaffari and Merav Parter.
\newblock {MST} in log-star rounds of congested clique.
\newblock In {\em Proceedings of the 2016 {ACM} Symposium on Principles of
  Distributed Computing, {PODC} 2016, Chicago, IL, USA, July 25-28, 2016},
  pages 19--28, 2016.

\bibitem[GP17]{GhaffariP17}
Mohsen Ghaffari and Merav Parter.
\newblock Near-optimal distributed {DFS} in planar graphs.
\newblock In {\em 31st International Symposium on Distributed Computing, {DISC}
  2017, October 16-20, 2017, Vienna, Austria}, pages 21:1--21:16, 2017.

\bibitem[HHW18]{HaeuplerHW18}
Bernhard Haeupler, D.~Ellis Hershkowitz, and David Wajc.
\newblock Round- and message-optimal distributed graph algorithms.
\newblock In {\em Proceedings of the 2018 {ACM} Symposium on Principles of
  Distributed Computing, {PODC} 2018, Egham, United Kingdom, July 23-27, 2018},
  pages 119--128, 2018.

\bibitem[HIZ16]{haeupler2016LCSwoEmbed}
Bernhard Haeupler, Taisuke Izumi, and Goran Zuzic.
\newblock Low-congestion shortcuts without embedding.
\newblock In {\em Proceedings of the 2016 ACM Symposium on Principles of
  Distributed Computing}, pages 451--460. ACM, 2016.

\bibitem[HL18]{HaeuplerL18}
Bernhard Haeupler and Jason Li.
\newblock Faster distributed shortest path approximations via shortcuts.
\newblock In {\em 32nd International Symposium on Distributed Computing, {DISC}
  2018, New Orleans, LA, USA, October 15-19, 2018}, pages 33:1--33:14, 2018.

\bibitem[HLZ18]{HaeuplerLZ18}
Bernhard Haeupler, Jason Li, and Goran Zuzic.
\newblock Minor excluded network families admit fast distributed algorithms.
\newblock In {\em Proceedings of the 2018 {ACM} Symposium on Principles of
  Distributed Computing, {PODC} 2018, Egham, United Kingdom, July 23-27, 2018},
  pages 465--474, 2018.

\bibitem[HPP{\etalchar{+}}15]{HegemanPPSS15}
James~W. Hegeman, Gopal Pandurangan, Sriram~V. Pemmaraju, Vivek~B. Sardeshmukh,
  and Michele Scquizzato.
\newblock Toward optimal bounds in the congested clique: Graph connectivity and
  {MST}.
\newblock In {\em Proceedings of the 2015 {ACM} Symposium on Principles of
  Distributed Computing, {PODC} 2015, Donostia-San Sebasti{\'{a}}n, Spain, July
  21 - 23, 2015}, pages 91--100, 2015.

\bibitem[JN18]{Jurdzinski018}
Tomasz Jurdzinski and Krzysztof Nowicki.
\newblock {MST} in \emph{O}(1) rounds of congested clique.
\newblock In {\em Proceedings of the Twenty-Ninth Annual {ACM-SIAM} Symposium
  on Discrete Algorithms, {SODA} 2018, New Orleans, LA, USA, January 7-10,
  2018}, pages 2620--2632, 2018.

\bibitem[KKOI19]{KitamuraKOI19}
Naoki Kitamura, Hirotaka Kitagawa, Yota Otachi, and Taisuke Izumi.
\newblock Low-congestion shortcut and graph parameters.
\newblock In Jukka Suomela, editor, {\em 33rd International Symposium on
  Distributed Computing, {DISC} 2019, October 14-18, 2019, Budapest, Hungary},
  volume 146 of {\em LIPIcs}, pages 25:1--25:17. Schloss Dagstuhl -
  Leibniz-Zentrum f{\"{u}}r Informatik, 2019.

\bibitem[LMR99]{leighton1999fast}
Tom Leighton, Bruce Maggs, and Andrea~W Richa.
\newblock Fast algorithms for finding {O}(congestion+ dilation) packet routing
  schedules.
\newblock {\em Combinatorica}, 19(3):375--401, 1999.

\bibitem[LP19]{LiP19}
Jason Li and Merav Parter.
\newblock Planar diameter via metric compression.
\newblock In {\em Proceedings of the 51st Annual {ACM} {SIGACT} Symposium on
  Theory of Computing, {STOC} 2019, Phoenix, AZ, USA, June 23-26, 2019}, pages
  152--163, 2019.

\bibitem[LPP01]{LotkerPP01}
Zvi Lotker, Boaz Patt{-}Shamir, and David Peleg.
\newblock Distributed {MST} for constant diameter graphs.
\newblock In {\em Proceedings of the Twentieth Annual {ACM} Symposium on
  Principles of Distributed Computing, {PODC} 2001, Newport, Rhode Island, USA,
  August 26-29, 2001}, pages 63--71, 2001.

\bibitem[LPP06]{LotkerPP06}
Zvi Lotker, Boaz Patt{-}Shamir, and David Peleg.
\newblock Distributed {MST} for constant diameter graphs.
\newblock {\em Distributed Comput.}, 18(6):453--460, 2006.

\bibitem[LPPSP03]{lotker2003mst}
Zvi Lotker, Elan Pavlov, Boaz Patt-Shamir, and David Peleg.
\newblock {M}{S}{T} construction in {O}({$\log log n$}) communication rounds.
\newblock In {\em the Proceedings of the Symposium on Parallel Algorithms and
  Architectures}, pages 94--100. ACM, 2003.

\bibitem[Now19]{NowickiMST20}
Krzysztof Nowicki.
\newblock A deterministic algorithm for the {MST} problem in constant rounds of
  congested clique.
\newblock {\em CoRR}, abs/1912.04239, 2019.

\bibitem[Pel00]{Peleg:2000}
David Peleg.
\newblock {\em Distributed Computing: A Locality-sensitive Approach}.
\newblock Society for Industrial and Applied Mathematics, Philadelphia, PA,
  USA, 2000.

\bibitem[SHK{\etalchar{+}}11]{SarmaHKKNPPW11}
Atish~Das Sarma, Stephan Holzer, Liah Kor, Amos Korman, Danupon Nanongkai,
  Gopal Pandurangan, David Peleg, and Roger Wattenhofer.
\newblock Distributed verification and hardness of distributed approximation.
\newblock In {\em Proceedings of the 43rd {ACM} Symposium on Theory of
  Computing, {STOC} 2011, San Jose, CA, USA, 6-8 June 2011}, pages 363--372,
  2011.

\bibitem[SHK{\etalchar{+}}12]{sarma2012distributed}
Atish~Das Sarma, Stephan Holzer, Liah Kor, Amos Korman, Danupon Nanongkai,
  Gopal Pandurangan, David Peleg, and Roger Wattenhofer.
\newblock Distributed verification and hardness of distributed approximation.
\newblock {\em SIAM Journal on Computing}, 41(5):1235--1265, 2012.

\end{thebibliography}

\end{document}